\documentclass[reqno]{amsart}

\usepackage{enumitem}
\usepackage{amssymb,amsmath,mathtools}
\usepackage{microtype}
\usepackage{bbm}
\usepackage{hyperref}
\usepackage{upgreek}

\newcommand*{\mailto}[1]{\href{mailto:#1}{\nolinkurl{#1}}}
\newcommand{\arxiv}[1]{\href{http://arxiv.org/abs/#1}{arXiv:#1}}

\newcommand{\msc}[1]{\href{http://www.ams.org/msc/msc2010.html?t=&s=#1}{#1}}

\newtheorem{theorem}{Theorem}[section]
\newtheorem{lemma}[theorem]{Lemma}
\newtheorem{corollary}[theorem]{Corollary}
\newtheorem{remark}[theorem]{Remark}

\newcommand{\R}{{\mathbb R}}
\newcommand{\N}{{\mathbb N}}

\newcommand{\C}{{\mathbb C}}

\newcommand{\id}{{\mathbbm{1}}}

\newcommand{\dom}{\mathrm{dom}}

\newcommand{\I}{\mathrm{i}}
\newcommand{\E}{\mathrm{e}}

\DeclareMathOperator{\re}{Re}
\DeclareMathOperator{\im}{Im}

\newcommand{\nn}{\nonumber}
\newcommand{\be}{\begin{equation}}
\newcommand{\ee}{\end{equation}}

\newcommand{\lam}{\lambda}

\newcommand{\inda}{\upalpha}
\newcommand{\indb}{\upbeta}

\numberwithin{equation}{section}

\makeatletter
\newcommand{\dlmf}[1]{%
\cite[%
 \def\nextitem{\def\nextitem{, }}%
 \@for \el:=#1\do{\nextitem\expandafter\dlmf@eq@href\el...\end}%
]{dlmf}%
}
\def\dlmf@eq@href#1.#2.#3.#4\end{%
  \href{http://dlmf.nist.gov/#1.#2.E#3}{(#1.#2.#3)}}
\makeatother

\begin{document}

\title[Dispersion Estimates]{Dispersion Estimates for the Discrete Laguerre Operator}

\author[A. Kostenko]{Aleksey Kostenko}
\address{Faculty of Mathematics\\ University of Vienna\\
Oskar-Morgenstern-Platz 1\\ 1090 Wien\\ Austria}
\email{\mailto{duzer80@gmail.com};\mailto{Oleksiy.Kostenko@univie.ac.at}}
\urladdr{\url{http://www.mat.univie.ac.at/~kostenko/}}

\author[G. Teschl]{Gerald Teschl}
\address{Faculty of Mathematics\\ University of Vienna\\
Oskar-Morgenstern-Platz 1\\ 1090 Wien\\ Austria\\ and International 
Erwin Schr\"odinger Institute for Mathematical Physics\\ 
Boltzmanngasse 9\\ 1090 Wien\\ Austria}
\email{\mailto{Gerald.Teschl@univie.ac.at}}
\urladdr{\url{http://www.mat.univie.ac.at/~gerald/}}

\thanks{{\it Research supported by the Austrian Science Fund (FWF) 
under Grant No.\ P26060}}
\thanks{Lett. Math. Phys. {\bf 106}, 545--555 (2016)}


\begin{abstract}
We derive an explicit expression for the kernel of the evolution group $\exp(-\mathrm{i} t H_0)$ of the discrete Laguerre operator $H_0$
(i.e.\ the Jacobi operator associated with the Laguerre polynomials) in terms of Jacobi polynomials.
Based on this expression we show that the norm of the evolution group acting from $\ell^1$ to $\ell^\infty$ is given by $(1+t^2)^{-1/2}$.
\end{abstract}

\maketitle

\section{Introduction}

We are concerned with the one-dimensional discrete Schr\"odinger equation
\begin{equation} \label{Schr}
  \I \dot \psi(t,n) = H_0 \psi(t,n), \quad  
  (t,n)\in\R\times \N_0,
\end{equation}
associated with the Laguerre operator
\be\label{eq:H0}
H_0 = \begin{pmatrix} 
1 & 1 & 0 & 0 & \cdots \\
1 & 3 & 2 & 0 & \cdots \\
0 & 2 & 5 & 3 & \cdots \\
0 & 0 & 3 & 7 & \cdots \\
\vdots &\vdots&\vdots&\vdots&\ddots
\end{pmatrix},
\ee
in $\ell^2(\N_0)$. Explicitly, $H_0 = (h_{n,m})_{n,m\in\N_0}$ with $h_{n,n} = 2n+1$, $h_{n,n+1} = h_{n+1,n} = n+1$ and $h_{n,m}=0$ whenever $|n-m|>1$.
Note that $h_{n,n} = h_{n-1,n}+h_{n,n+1}$.
It is a special case of a self-adjoint Jacobi operator whose generalized eigenfunctions are precisely the Laguerre polynomials explaining our name.

This operator appeared recently in the study of radial waves in $(2+1)$-dimensional noncommutative scalar field theory \cite{a06,gms}
and has attracted further interest in \cite{cfw03,ks15a,ks15b,ks15c}. More precisely, \eqref{Schr} is the linear part in the nonlinear Schr\"odinger equation (NLS)
\be\label{eq:nls}
\I \dot \psi(t,n) = H_0 \psi(t,n) - |\psi(t,n)|^{2\sigma}\psi(t,n),\quad \sigma\in\N,\quad (t,n)\in \R_+\times \N_0,
\ee
investigated in the recent work of Krueger and Soffer \cite{ks15a,ks15b,ks15c}. Also $H_0$ appeared in the discrete nonlinear Klein--Gordon equation (NLKG) \cite{cfw03,gms}.
In turn, the dynamics of noncommutative solitons in the context of noncommutative field theory (see, e.g., \cite{bbvy, gms, lec} for reviews) can be reduced to the study of discrete NLKG and NLS equations. In contrast to asymptotic metastability for the NLKG solitons conjectured in \cite{cfw03}, it is expected that the NLS solitons are asymptotically stable (see  \cite[\S 8]{ks15c}). In this connection, let us emphasize that dispersive estimates for the linear part \eqref{Schr} as well as for its perturbations are an important ingredient in the standard procedure for the proof of asymptotic stability for nonlinear PDEs (see \cite{sw1,sw2,bs,bp}). In fact, the case of Jacobi operators with asymptotically constant coefficients has already attracted a lot of attention and we refer to \cite{EHT,EKT} and the references therein (see also \cite{KT} for discrete Dirac-type operators).

To formulate our results we recall the weighted spaces $\ell^p_{\sigma}=\ell^p_{\sigma}(\N_0)$,
$\sigma\in\R$, associated with the norm
\begin{equation*}
   \Vert u\Vert_{\ell^p_{\sigma}}= \begin{cases} \left( \sum_{n\in\N_0} (1+n)^{p\sigma} |u(n)|^p\right)^{1/p}, & \quad p\in[1,\infty),\\
   \sup_{n\in\N_0} (1+n)^{\sigma} |u(n)|, & \quad p=\infty. \end{cases}
\end{equation*}
Of course, the case $\sigma=0$ corresponds to the usual $\ell^p_0=\ell^p$ spaces without weight. Then, in \cite[Theorem 2]{ks15b} it was shown that
\begin{equation}\label{decayks}
\Vert \E^{-\I tH_0}\Vert_{\ell^1_\sigma\to \ell^\infty_{-\sigma}}=\mathcal{O}(t^{-1}),\quad t\to\infty.
\end{equation}
for $\sigma\ge 3$. This is in contrast to the case of the discrete Laplacian $\Delta$, where one has (cf.\ e.g.\ \cite{EKT})
\begin{equation}
\Vert \E^{-\I t\Delta}\Vert_{\ell^1\to \ell^\infty}=\mathcal{O}(t^{-1/3}),\quad t\to\infty.
\end{equation}

The purpose of the present note is to improve \eqref{decayks} by showing that the weights are not
necessary, that is, it holds for $\sigma\ge 0$. Even more, we are able to compute this norm explicitly:

\begin{theorem}\label{thm:decay}
The following equality
\be\label{eq:decay}
\|\E^{-\I tH_0}\|_{\ell^1\to \ell^\infty} = \frac{1}{\sqrt{1+t^2}},\quad t\in\R,
\ee
holds.
\end{theorem}

This result in turn is based on the following explicit expression for the kernel of the evolution group $\E^{-\I tH_0}$ given in terms of Jacobi polynomials (see \cite{aar,sz} for the definition and basic properties):

\begin{theorem}\label{thm:explicit}
Let $n$, $m\in\N_0$ be such that $n\le m$. Then
\begin{align}
\E^{-\I tH_0}(n,m)
	& = \frac {1}{1+\I t}\left(\frac{\I+t}{\I - t}\right)^n \left(\frac{t}{\I-t}\right)^{m-n} P_n^{(m-n,0)}\left( \frac{1-t^2}{1+t^2}\right),\label{eq:explicit}
\end{align}
where
\be\label{eq:Jacpol}
P_n^{(m-n,0)}(z)=\sum_{k=0}^n \binom{m}{n-k} \binom{n}{k} \left(\frac{z-1}{2}\right)^k\left(\frac{z+1}{2}\right)^{n-k}
\ee
 is the Jacobi polynomial.
\end{theorem} 

As it was already mentioned, the understanding of the dynamics of \eqref{Schr}, which is the linear part of \eqref{eq:nls}, is of crucial importance in the study of \eqref{eq:nls} and the NLKG equations. Furthermore, the understanding of the free evolution \eqref{Schr} is a necessary prerequisite for a successful development of scattering theory.

The proof of Theorem \ref{thm:decay} and Theorem \ref{thm:explicit} is given in the next section and it is based on the fact that every element of the kernel of $\E^{-\I tH_0}$ is a Laplace transform of a product of two Laguerre polynomials (Lemma \ref{lem:exponent}). Now notice that the required decay estimate will follow once we have a uniform estimate for the Jacobi polynomials $P_n^{(m-n,0)}(\cdot)$ on the segment of orthogonality $[-1,1]$. Unfortunately, the standard estimate (see, e.g., \cite[Theorem 7.32.1]{sz}) only gives 
\be\label{eq:sz_est}
\max_{x\in [-1,1]}\big|P_n^{(m-n,0)}(x)\big| = \binom{m}{n},\quad n\le m,\ n,m\in\N,
\ee
which is clearly insufficient for our purposes. It is somewhat surprising that the required estimate follows from the unitarity of the so-called Wigner $d$-matrix (Jacobi polynomials appear as matrix elements for the irreducible representations of ${\rm SU}(2)$, see \cite{aar,vil}). Even more, to the best of our knowledge, the analytic proof of this estimate was obtained only recently by Haagerup and Schlichtkrull in \cite{haa}. 

Let us also mention one more dispersive estimate which follows from the Haagerup--Schlichtkrull inequality for Jacobi polynomials  \cite[Theorem 1.1]{haa} (see \eqref{eq:haagerup} below).

\begin{theorem}\label{thm:decay_nm}
There is a constant $C\le  2\sqrt[4]{42}$ such that the following inequality 
\begin{align}\label{eq:decay_nm}
\big|\E^{-\I tH_0}(n,m)\big| \le \frac{C}{t^{1/2} (n+m+1)^{1/4}},\quad n,m\in\N_0,
\end{align}
holds for all $t>0$.
\end{theorem} 

Note that the estimate \eqref{eq:decay_nm} does not provide an optimal decay rate in $t$, however, it gives an additional decay of the coefficients of the kernel in $n$ and $m$.  
Let us mention that the Haagerup--Schlichtkrull inequality \eqref{eq:haagerup} was derived in order to obtain uniform bounds on a complete set of matrix coefficients for the irreducible representations of ${\rm SU}(2)$ with a decay rate $d^{-1/4}$ in the dimension $d$ of the representation. Furthermore, the Bernstein (see \eqref{eq:bernst}) and the Haagerup--Schlichtkrull estimates were used in \cite{ldl} and \cite{hdl}, respectively, for establishing the absence of the approximation property of Haagerup and Kraus \cite{hk} for ${\rm SL}(3,\R)$ and ${\rm Sp}(2,\R)$. 

On the other hand, in the follow-up paper \cite{kt16}, we investigate the decay estimate for generalized Laguerre operators $H_\alpha$ (tri-diagonal matrices associated with generalized Laguerre polynomials $L_n^{(\alpha)}$, see \cite{sz}), where the coefficient $\alpha$ can be seen as a measure of the delocalization of the field configuration and it is related to the planar angular momentum \cite{a13} ($\alpha=0$ corresponds to radial waves in $(2+1)$-dimensional noncommutative scalar field theory).  
It turned out that the optimal dispersive decay estimate leads to new Bernstein-type inequalities for Jacobi polynomials. 
All these connections are mathematically very appealing and we hope that the present note will stipulate further research in this direction.

To end this section, let us briefly outline the content of the paper. Before proving the main result, we collect the basic spectral properties of the operator $H_0$ in Theorem \ref{thm:spH_0} and also present its proof based on spectral theory of Jacobi operators. Next, in Lemma \ref{lem:exponent} we represent the kernel of the evolution group $\E^{-\I tH_0}$ by means of the Laguerre polynomials and then prove our main results Theorem \ref{thm:explicit}, Theorem \ref{thm:decay}, and Theorem \ref{thm:decay_nm}. Finally, in Lemma \ref{lem:conv} we present another representation of the kernel of $\E^{-\I tH_0}$, which might be of independent interest. In particular, it allows to obtain a simple proof of \eqref{decayks} for $\sigma\ge 1/2$.

\section{Proof of the main results}

We start with a precise definition of the operator $H_0$. Let $D:\ell^2(\N_0) \to \ell^2(\N_0)$ be the multiplication operator given by
\be\label{eq:D}
(Du)_n = (n+1)u_n, \quad u \in \dom(D)=\ell_1^2(\N_0).
\ee
For a sequence $u=\{u_n\}_{n\ge 0}$ we define the difference expression $\tau: u\mapsto \tau u$ by setting
\be\label{eq:tau}
(\tau u)_n : = \begin{cases} u_0 + u_1, & n=0,\\ nu_{n-1} + (2n+1)u_n + (n+1) u_{n+1}, & n\ge 1. \end{cases}
\ee
Then the operator $H_0$ associated with the Jacobi matrix \eqref{eq:H0} is defined by
\begin{align}\begin{split}
H_0: \begin{array}[t]{lcl} \mathcal{D}_{\max} &\to& \ell^2(\N_0) \\ u &\mapsto& \tau u, \end{array}
\end{split}\end{align}
where $ \mathcal{D}_{\max} = \{u\in \ell^2(\N_0):\, \tau u\in \ell^2(\N_0)\}$.
Note that $\ell_1^2(\N_0)\subset \mathcal{D}_{\max}$, however, simple examples (take $u=\{{(-1)^n}/{(n+1)}\}_{n\ge 0}$) show that the inclusion is strict.

The spectral properties of $H_0$ were derived in \cite{cfw03, ks15b}, however, without using the well-developed spectral
theory for Jacobi operators \cite{tjac}. We collected them in the following theorem and give a short proof using this connection.

\begin{theorem}\label{thm:spH_0}
\begin{itemize}
\item[(i)] The operator $H_0$ is a positive self-adjoint operator.
\item[(ii)] The Weyl function and the corresponding spectral measure are given by
\be
m_0(z) = \E^{-z}E_1(-z) = \int_0^{+\infty} \frac{\E^{-\lambda}}{\lambda-z}\, d\lambda ,\quad d\rho(\lambda) = \id_{\R_+}(\lambda)\E^{-\lambda}d\lambda,
\ee
where $E_1$ denotes the principal value of the exponential integral \dlmf{6.2.1}.
\item[(iii)] The spectrum of $H_0$ is purely absolutely continuous and coincides with $\R_+=[0,\infty)$.
\end{itemize}
\end{theorem}

\begin{proof}
(i) Self-adjointness clearly follows from the Carleman test (see, e.g., \cite{akh}, \cite[(2.165)]{tjac}) since $\sum_{n\ge 0} (n+1)^{-1} = \infty$.
Nonnegativity follows from the following representation of the matrix \eqref{eq:H0}
\be\label{eq:H0factor}
H_0 = (I+U) D (I+U^*),
\ee 
where $U:(u_0,u_1,u_2,\dots)\mapsto (0,u_0,u_1,u_2,\dots)$ is the forward shift on $\ell^2(\N_0)$ and $U^*:(u_0,u_1,u_2,\dots)\mapsto (u_1,u_2,u_3,\dots)$ is its adjoint, the backward shift operator. Moreover, using this factorization, it is not difficult to check that the kernel of $H_0$ is trivial, $\ker(H_0) = \{0\}$.

(ii) Notice that the polynomials of the first kind for $H_0$ are given by
\be
P_n(z) = (-1)^{n} L_n(z),\quad n\in \N_0,
\ee
where
\be\label{eq:laguerpol}
L_n(z) = \frac{1}{n!}\left(\frac{d}{dz} -1\right)^n z^n = \sum_{k=0}^n  \binom{n}{k}\frac{(-z)^k}{k!},\quad n\in\N_0,
\ee
are the Laguerre polynomials. Indeed (see, e.g., \cite[Chapter~V]{sz}), they satisfy the following recursion relations
\be\label{eq:lagrecurs}
\begin{split}
L_0(z) - L_1(z) & = z L_0(z),\\
-nL_{n-1}(z) + (2n+1)L_n(z) - (n+1) L_{n+1}(z) & = z L_n(z),\quad n\ge 1,
\end{split}
\ee
and the orthogonality relations
\be\label{eq:lagorth}
\int_{[0,\infty)} L_n(\lambda) L_k(\lambda)  \E^{-\lambda}\, d\lambda = \delta_{nk},\quad n,k\in\N_0.
\ee
Therefore, \eqref{eq:lagorth} and (i) imply that $d\rho(\lambda) = \id_{\R_+}(\lambda)\E^{-\lambda}d\lambda$ is the spectral measure of $H_0$, that is, $H_0$ is unitarily equivalent to a multiplication operator in $L^2(\R_+,d\rho)$ (cf.\ e.g.\ \cite[Theorem 2.12]{tjac}). 
It remains to note that the corresponding Weyl function is the Stieltjes transform of the measure $d\rho$ (cf.\ e.g.\ \cite[Chapter~2]{tjac}). 

 (iii) The claim immediately follows from (ii). Indeed, as it was already mentioned, $H_0$ is unitarily equivalent to a multiplication operator $\tilde{H}$ acting in $L^2(\R_+,d\rho)$,
\begin{align*}\begin{split}
\tilde{H}: \begin{array}[t]{lcl} \dom(\tilde{H}) &\to& L^2(\R_+,d\rho) \\ f(\lam) &\mapsto& \lam f(\lam), \end{array}
\end{split}\end{align*}
where $\dom(\tilde{H}) = L^2(\R_+;\lam^2d\rho)$
Hence the spectra as well as the spectral types of both operators coincide \cite[Eq. (2.106)]{tjac}). It remains to note that the spectrum of $\tilde{H}$ is purely absolutely continuous and coincides with $\R_+=[0,\infty)$. 
\end{proof}

\begin{remark}
Note that
\[
m_0(-x) \uparrow +\infty\quad \text{as}\quad x\downarrow 0.
\]
\end{remark}

Next, let us define the polynomials of the second kind (see \cite{akh,tjac})
\[
Q_n(z) = (-1)^{n}\int_0^\infty \frac{L_n(z)-L_n(\lam)}{z-\lam} \E^{-\lam}d\lam,
\]
which satisfy $(\tau u)_n = z u_n$ for all $n\ge 1$ with $Q_{0}\equiv 0$, $Q_1\equiv 1$. It follows from \eqref{eq:lagrecurs} that
\[
L_{n}(z) Q_{n+1}(z) + L_{n+1}(z)Q_n(z) = \frac{(-1)^{n}}{n+1},\quad n\in\N_0.
\]
Moreover, for all $z\in \C\setminus\R_+$ the linear combination
\be\label{eq:weylsol}
\Psi_n(z):= Q_n(z) + m_0(z) P_n(z),\quad n\in\N_0, 
\ee
also known as Weyl solution in the Jacobi operator context, satisfies $\{\Psi_n(z)\}_{n\ge 0} \in \ell^2(\N_0)$. Therefore, the resolvent of $H_0$ is given by
\be\label{eq:H0resolv}
G(z;n,m)= \langle (H_0 - z)^{-1}\delta_n ,\delta_m \rangle  = \begin{cases} (-1)^{n}L_n(z) \Psi_m(z), & n\le m,\\ (-1)^m L_m(z) \Psi_n(z), & n\ge m. \end{cases}
\ee
Notice that $G(z;0,0)= \langle (H_0 - z)^{-1}\delta_0 ,\delta_0 \rangle  = m_0(z)$.

The next result provides an integral representation of the operator $\E^{-\I tH_0}$ in terms of the Laguerre polynomials.

\begin{lemma}\label{lem:exponent}
The kernel of the operator $\E^{-\I tH_0}$ is given by
\begin{align}\label{PP}
\E^{-\I tH_0}(n,m)
	=(-1)^{n+m} \int_{0}^{\infty}\E^{-\I t \lambda} L_n(\lambda) L_m(\lambda) \E^{-\lambda} d\lambda,\quad n,m\in\N_0. 	 
\end{align}
\end{lemma}

\begin{proof}
From Stone's formula \cite{tschroe} we know
\[
\E^{-\I tH_0}(n,m)=\frac 1{2\pi \I}\int_{0}^{\infty}\E^{-\I t\lambda}
	  \left[G(\lambda+\I0;n,m)-G(\lambda-\I 0;n,m)\right] d\lambda.
\]
Since
\[
G(\lambda\pm \I0;n,m) = (-1)^n L_n(\lambda) \big(Q_m(\lambda) \pm  m_0(\lambda+\I 0) (-1)^m L_m(\lambda)\big),\quad  \lambda>0,
\]
if $n\le m$, it remains to note that $L_n(\lambda)^* = L_n(\lambda)$, $Q_m(\lambda)^* = Q_m(\lambda)$, and, moreover, $ \im m_0(\lambda+\I 0) = \pi\E^{-\lambda}$ for all $\lambda>0$.
\end{proof}

It follows from \eqref{PP} that every element of the kernel of the operator $\E^{-\I tH_0}$ is the Laplace transform of a product of the corresponding Laguerre polynomials and hence one can compute them explicitly:

\begin{proof}[Proof of Theorem \ref{thm:explicit}]
Recall that (see \cite[(4.11.35)]{erd})
\[
\int_0^\infty \E^{-p\lambda} L_n(\lambda)L_m(\lambda) d\lambda = \binom{n+m}{n} \frac{(p-1)^{n+m}}{p^{n+m+1}} {}_2 F_1\left(-n,-m;-n-m; \frac{p(p-2)}{(p-1)^2} \right),
\]
whenever $\re(p)> 0$ with ${}_2 F_1$ is the hypergeometric function (see \cite[Chapter 15]{dlmf}) 
\begin{align}
{}_2 F_1 (a,b;c; z) = \sum_{k=0}^\infty \frac{(a)_k(b)_k}{(c)_k}\frac{z^k}{k!},
\end{align}
where $(x)_k = \prod_{j=0}^{k-1} (x+j)$. 
Setting $p= 1+\I t$, we get
\[
\E^{-\I tH_0}(n,m) = \frac {1}{1+\I t}\left(\frac{-\I t}{1+\I t}\right)^{n+m} {n+m\choose n} {}_2 F_1\left(-n,-m;-n-m;1+\frac{1}{t^2}\right).
\]
Finally we recall the connection with the Jacobi polynomials \cite[(4.22.1)]{sz}
\[
P_k^{(\inda,\indb)}(z)= {2k+\inda+\indb\choose n} \left(\frac{z-1}{2} \right)^k \,{}_2F_1\left(-k,-k-\inda;-2k-\inda-\indb;\frac{2}{1-z}\right),
\]
which establishes \eqref{eq:explicit}.
\end{proof}

We note some special cases:
\begin{corollary}\label{cor:cases}{\ }
\begin{itemize}
\item[(i)] In the case $n=0$ we have
\[ 
\E^{-\I tH_0}(0,m) =\E^{-\I tH_0}(m,0) = \frac {1}{1+\I t}\left(\frac{t}{\I-t}\right)^m,\quad m\in\N_0.
\]
\item[(ii)] In the case $m=1$ we have
\[
\E^{-\I tH_0}(1,m) =\E^{-\I tH_0}(m,1) = \frac {1}{1+\I t} \left(\frac{t}{\I - t} \right)^{m+1} \frac{t^2-m}{t^2},\quad m\in\N_0.
\]
\item[(iii)]
In the case $n=m$ we have
\[
\E^{-\I tH_0}(n,n) = \frac {1}{1+\I t}\left(\frac{\I+t}{\I-t}\right)^n P_n\left( \frac{1-t^2}{1+t^2}\right), \quad n\in\N_0,
\]
where 
\[
P_n(z) = \frac{1}{2^n\Gamma(n+1)}\frac{d^n}{dz^n} (z^2-1)^{n}
\]
are the Legendre polynomials \cite{sz}.
\end{itemize}
\end{corollary}

\begin{proof}
Just observe
\[
P_0^{(m,0)} = 1, \quad P_1^{(m-1,0)} = \frac{m-1+(m+1) x}{2},\quad P_n^{(0,0)}(z)=P_n(z).\qedhere
\]
\end{proof}

In all the above cases we have $|\E^{-\I tH_0}(\cdot,\cdot)| \le   t^{-1}$ for all $t> 0$ and $|\E^{-\I tH_0}(\cdot,\cdot)| \sim  t^{-1}$ as $|t|\to \infty$
(recall the well-known estimate $\max_{x\in[-1,1]}|P_n(x)| = 1$; see also \eqref{eq:sz_est}) and our final aim is to establish this estimate for all cases.
In this respect we also remark that using $P_n^{(m-n,0)}(-1)= (-1)^n$ we get
\[
\E^{-\I tH_0}(n,m)= \frac{1}{\I t} +O(t^{-2}),\quad t\to \infty,
\]
but the error is not uniform since $\frac{d}{dz} P_n^{(m-n,0)}(-1)=(-1)^n \frac{n(m-1)}{2}$.

\begin{proof}[Proof of Theorem \ref{thm:decay}]
Using Corollary \ref{cor:cases} (i), we get
\[
\|\E^{-\I tH_0}\|_{\ell_1\to \ell_\infty} \ge \big|\E^{-\I tH_0}(0,0)\big| = \frac{1}{\sqrt{1+t^2}},\quad t\in\R.
\]
The converse inequality  follows from the estimate \cite[formula (20)]{haa}
\be\label{eq:g-ineq}
|g_n^{(\inda,\indb)}(x)| \le \left(\frac{(n+1)(n+\inda+\indb+1)}{(n+\inda+1)(n+\indb+1)}\right)^{1/4},
\ee
where
\[
g_n^{(\inda,\indb)}(x) = \left(\frac{\Gamma(n+1)\Gamma(n+\inda+\indb+1)}{\Gamma(n+\inda+1)\Gamma(n+\indb+1)}\right)^{1/2}
\left(\frac{1-x}{2}\right)^{{\inda}/{2}}\left(\frac{1+x}{2}\right)^{\indb/2} P_n^{(\inda,\indb)}(x).
\]
The inequality \eqref{eq:g-ineq} holds true for all $x\in [-1,1]$ and $\inda$, $\indb\in\N_0$. In our case this reduces to
\[
\big|g_n^{(m-n,0)}(x)\big| \le \left(\frac{(n+1)(m+1)}{(m+1)(n+1)}\right)^{1/4} =1
\]
and the claim follows upon observing
\be\label{eq:fin}
\E^{-\I tH_0}(n,m)
	= \frac {1}{1+\I t}\left(\frac{t+\I}{t-\I}\right)^{\frac{m+n}{2}} g_n^{(m-n,0)}\left( \frac{1-t^2}{1+t^2}\right). \qedhere
\ee
\end{proof}

\begin{proof}[Proof of Theorem \ref{thm:decay_nm}] The proof is based on the following inequality (see \cite[Theorem 1.1]{haa})
\be\label{eq:haagerup}
|(1-x^2)^{1/4} g_n^{(\inda,\indb)}(x)| \le \frac{C}{(2n+\inda+\indb+1)^{1/4}},\quad n\in\N_0,
\ee
which holds with $C=2\sqrt[4]{168}$ for all $x\in [-1,1]$ and $\inda$, $\indb\in \N_0$ (see p.235 in \cite{haa} and also Lemma 4.3 there). 
Now it suffices to note that 
\[
(1-x^2)^{1/4} = \sqrt{\frac{2t}{1+t^2}},\quad x=\frac{1-t^2}{1+t^2},
\]
and then using \eqref{eq:fin}, we arrive at \eqref{eq:decay_nm}.
\end{proof}

\begin{remark}\label{rem:bernstein}
The decay rate $1/4$ in \eqref{eq:haagerup} is optimal as $\inda$ and $\indb$ tend to infinity (see \cite[Remark 4.4]{haa}). However, when $\inda$ and $\indb$ are fixed, for example, if $\inda=0$ and $\indb=0$, then the classical Bernstein inequality \cite[Theorem 7.3.3]{sz} (see also \cite{ant})
\be\label{eq:bernst}
(1-x^2)^{1/4} |P_n(x)| \le \frac{2}{\sqrt{\pi(2n+1)}},\quad n\in\N_0,\ \ x\in [-1,1],
\ee
together with Corollary~\ref{cor:cases} (iii) implies
\[
|\E^{-\I tH_0}(n,n)| \le \frac{1}{\sqrt{\pi t (n+1/2)}},\quad n\in\N_0.
\]
\end{remark}

Finally, we mention another representation of the kernel of $\E^{\I tH_0}$ which might be of independent interest.

\begin{lemma}\label{lem:conv}
Let
\be\label{eq:Fn}
F_n(t) = \frac {2}{1+2\I t}\left(\frac{1 - 2\I t}{1+2\I t}\right)^{n},\quad n\in\N_0.
\ee
Then
\begin{align}
\E^{-\I tH_0}(n,m) =& (-1)^{n+m} \big(F_n * F_m)(t) \nn\\
&= \frac{(-1)^{n+m}}{2\pi}  \int_{\R} F_n(s) F_m(t-s) ds,\quad t>0.\label{eq:explicit_b}
\end{align}
\end{lemma}

\begin{proof}
Using \cite[(4.11.31)]{erd} we compute
\[
\int_0^\infty \E^{-\I t\lambda} L_n(\lambda)\E^{-\lambda/2} d\lambda = \frac {2}{1+2\I t}\left(\frac{1 - 2\I t}{1+2\I t}\right)^{n} = F_n(t),\quad n\in\N_0.
\]
Noting that the Fourier transform of a product of two functions is equal to the convolution of their Fourier transforms and using \eqref{PP}, we end up with \eqref{eq:explicit_b}.  
\end{proof}

\begin{remark}
Noting that
\[
F_{n+1}(t) = \frac{2}{1+2\I t}F_n(t) - F_n(t),\quad (F_0 * F_n)(t) = \frac{1}{1+\I t} \left( \frac{\I t}{1+\I t}\right)^n,\quad n\in\N_0,
\]
and then estimating the convolution, one can show by using induction that 
\[
|\E^{-\I tH_0}(n,m)| \le \frac{1+|m-n|}{\sqrt{1+t^2}},\quad t\in\R.
\]
\end{remark}

{\bf Acknowledgments.}
We are indebted to August Krueger for drawing our attention to this problem and Avy Soffer for useful discussions. We thank the anonymous  referees for their careful reading of our manuscript and critical comments.


\end{document}